%% file: main.tex
\documentclass[a4paper,UKenglish,cleveref,thm-restate]{lipics-v2021}
\usepackage{bm}
\usepackage{todonotes}
\usepackage{mathtools}
\usepackage{arydshln}

\nolinenumbers

\newcommand{\R}{\mathbb R}
\newcommand{\Z}{\mathbb Z}
\newcommand{\veczero}{\mathbf0}
\newcommand{\vecone}{\mathbf1}
\newcommand{\matzero}{\bm O}
\newcommand{\nfold}{\textsc{n-fold}}
\renewcommand{\O}{\mathcal O}
\DeclareMathOperator{\conv}{conv}

\makeatletter
\AddToHook{cmd/appendix/before}{\def\cref@section@alias{appendix}\def\cref@subsection@alias{appendix}}
\makeatother

\crefname{lemma}{Lemma}{Lemmas}
\Crefname{lemma}{Lemma}{Lemmas}
\crefname{proposition}{Proposition}{Propositions}
\Crefname{proposition}{Proposition}{Propositions}
\crefname{theorem}{Theorem}{Theorems}
\Crefname{theorem}{Theorem}{Theorems}
\crefname{section}{Section}{Sections}
\Crefname{section}{Section}{Sections}
\crefname{appendix}{Appendix}{Appendices}
\Crefname{appendix}{Appendix}{Appendices}

\title{4-Block Integer Programming is in FPT}

\titlerunning{4-Block IP is in FPT}

\author{Martin Koutecký}{Charles University, Czechia}{koutecky@iuuk.mff.cuni.cz}{https://orcid.org/0000-0002-7846-0053}{}

\author{Alexandra Lassota}{Eindhoven University of Technology, Netherlands}{a.a.lassota@tue.nl}{https://orcid.org/0000-0001-6215-066X}{}

\author{Koen Ligthart}{Eindhoven University of Technology, Netherlands}{k.m.ligthart@tue.nl}{https://orcid.org/0009-0004-6823-5225}{}

\authorrunning{M. Koutecký, A. Lassota and K. Ligthart}

\Copyright{Martin Koutecký, Alexandra Lassota and Koen Ligthart} 



\hideLIPIcs

\begin{document}

\maketitle

\begin{abstract}
\input{sections/abstract}
\end{abstract}

\input{sections/introduction}

\input{sections/convex-extensibility}

\input{sections/algorithm}

\section*{AI Statement}
GPT-6-Astra Pro was used to develop and prove \cref{lemma:2d-dilation}, \cref{lemma:non-vertices-are-decomposable}, and \cref{proposition:bounded-fractionality-is-insufficient}. The application in \cref{sec:applications} was found with OpenAI Codex, and \cref{sec:applications} was drafted with Claude. Apart from these, the paper is written by the authors without the use of LLMs. The authors take full responsibility for its correctness and for all other contents of the paper.

\bibliography{bib}

\appendix

\input{sections/appendix}

\input{sections/applications}

\end{document}

%% file: sections/abstract.tex
Integer programming is a fundamental and important NP-hard problem. This motivated extensive efforts in studying several tractable subclasses. One of the most important, unresolved complexity questions is the parameterized complexity of 4-block IPs, a natural class characterized by having a diagonal matrix with small blocks after deleting few rows and columns. Over the years, significant progress has been made in improving algorithms for 4-block IPs, but the question whether such IPs can be solved in FPT time, parameterized by the block dimensions and largest matrix coefficient, has remained open. This question is repeatedly highlighted, most recently by Koutecký [IPEC 2025] and by Eisenbrand and Rothvoss [SODA 2026]. We resolve this question in the positive by providing an FPT time algorithm. Our algorithm can optimize non-linear, separable convex objective functions, and can be extended to broader classes of constraint matrices (such as tree-fold or multi-stage) and allows appending few ``global'' columns to it, and it allows coefficients unbounded by the parameters in those columns. It is known that tractability cannot be extended further in any of those directions.
The key structural property that we establish is that a function $f\colon\mathbb Z^n\to\mathbb R$ that is integer midpoint convex, i.e., $f(x)\le\tfrac12f(x-p)+\tfrac12f(x+p)$ for all $x,p\in\mathbb Z^n$, can be extended to a convex function on the set $2d\mathbb Z^n\cap L$ if $L$ is a linear subspace of dimension $d$. This closes the gap in a recent work by Ligthart [arXiv 2606.30330, 2026], which allows us to extend the previous algorithm that solves 4-block integer programs with a single global variable to incorporate any parameterized number of global variables. 
As an example, we apply our algorithm to solve an open problem posed by Boehmer and Koana [ACM Trans. Algorithms, 2024] in Max-Min Fair Matching.

%% file: sections/introduction.tex
\section{Introduction}
Solving integer programs (IPs) efficiently is one of the most fundamental problems in theoretical computer science. It is a central and powerful modeling tool in discrete optimization and has become a state-of-the-art solution framework in many application domains, e.g., AI planning~\cite{BrielVossenKambhampati05,VossenBallLotemNau99}, vehicle routing~\cite{Toth01}, process scheduling~\cite{FloudasLin05}, and packing~\cite{LodiMM02}. In standard form, the goal is to solve
\begin{align*}
\min\{c^\top x \mid Ax=b,\; l\le x\le u,\; x\in\mathbb{Z}^n\},
\end{align*}
where $c\in\mathbb{Z}^n$ is the objective coefficient vector, $A\in\mathbb{Z}^{m\times n}$ is the constraint matrix, $b\in\mathbb{Z}^m$ is the right-hand side vector, and $l,u\in\mathbb{Z}^n$ are lower and upper bounds on the variables.

Despite playing such a central role, solving the problem is NP-hard in general, already appearing as one of Karp's 21 NP-hard problems~\cite{Kar72}. To overcome this daunting complexity, tractable subclasses of integer programs have been intensively studied.
In the 1960s, Hoffman and Kruskal~\cite{hoffman2010integral} showed that IPs with $A$ totally unimodular can be solved by solving their LP relaxation, which gives a polynomial algorithm.
The natural generalization of this class to $\delta$-modular IPs, i.e., those where the largest subdeterminant of $A$ is bounded by $\delta$ in absolute value, has seen much progress recently: after the polynomial-time algorithm for the bimodular case $\delta=2$~\cite{ArtmannWZ:2017,VeselovC:2009}, progress has been made for $\delta\in\{3,4\}$ under a strict modularity assumption~\cite{NageleSZ22,NageleNSZ23}, for matrices with two nonzeros per row~\cite{FioriniJWY21}, and on the structure of $\delta$-modular polyhedra~\cite{GribanovMS24}, but the existence of an algorithm running in polynomial time for any fixed $\delta$ remains an important open problem.
In the 1980s, Lenstra~\cite{DBLP:journals/mor/Lenstra83} showed that IPs with few variables are tractable, and subsequent work improved the superexponential dependency on the dimension, so far culminating with a recent breakthrough by Reis and Rothvoss~\cite{DBLP:conf/focs/ReisR23}.
A major open problem is whether the naive complexity lower bound of $2^{\O(n)} L^{\O(1)}$ is achievable, with $L$ the length of the input in binary.

A third important class that emerged starting around 2000 are so-called block-structured integer programs; for a brief history, see the recent survey of Koutecký~\cite{Koutecky25}. In such IPs, the constraint matrix has a certain sparse structure. The most prominent of such matrices are \emph{$n$-fold} and \emph{2-stage stochastic} matrices, whose constraint matrices respectively look as follows:

\begin{equation}
\begin{aligned}
A_{\text{n-fold}} &=
\begin{bmatrix}
B_1&\cdots&B_n\\
D_1&&\\
&\ddots&\\
&&D_n
\end{bmatrix},
\qquad\qquad
A_{\text{2-stage stochastic}} =
\begin{bmatrix}
C_1&D_1&&\\
\vdots&&\ddots&\\
C_n&&&D_n
\end{bmatrix},
\end{aligned}
\label{eq:nfoldmatrix}
\end{equation}
where all ommitted entries are zero, and where the matrices are of dimensions $B_i\in\{-\Delta,-\Delta+1,\dots,\Delta\}^{r\times t}$, $C_i\in\mathbb{Z}^{s\times p}$, and $D_i\in\{-\Delta,-\Delta+1,\dots,\Delta\}^{s\times t}$ for $i\in[n]=\{1,2,\dots,n\}$.

These IPs arise naturally in a broad range of problems, see~\cite{ChenMYZ17,DeLoera2008,DBLP:journals/disopt/GavenciakKK22,JansenKMR22,KnopK18,KnopK22,KnopKLMO21,KnopKM20,KnopKM20b,KouteckyZ20} for just a few examples. Developing increasingly efficient algorithms for solving block-structured IPs has become a joint effort across multiple fields, leading to numerous papers that improve the running time gradually or allow for broader generalizations, see~\cite{AschenbrennerH07,CslovjecsekEHRW21,CslovjecsekEPVW21,DBLP:journals/theoretics/CslovjecsekKLPP25,DeLoera2008,EisenbrandHK18,monster,DBLP:conf/soda/EisenbrandR26,HemmeckeOR13,JansenLR20,DBLP:journals/mp/Klein22,KouteckyLO18,ligthart2026valuefunctionsseparableconvex}.
The natural generalization of both $n$-fold and 2-stage IPs is to allow \emph{both} ``global'' variables (like 2-stage IPs) and ``global'' constraints (like $n$-fold IPs).
This led Hemmecke, Köppe, and Weismantel~\cite{DBLP:conf/ipco/HemmeckeKW10} to study 4-block $n$-fold IPs (from hereon simply 4-block IPs) and give an algorithm running in time $n^{g(p,r,s,t,\Delta)}$.
In the language of parameterized complexity, this is a so-called \emph{slice-wise polynomial}, or \emph{XP}, algorithm, whose existence always raises the question whether the generally faster \emph{fixed-parameter tractable} (\emph{FPT}) algorithm with complexity $g(p,r,s,t,\Delta)n^{\O(1)}$ exists.
This question was raised explicitly multiple times~\cite{FPTNews16,DBLP:journals/disopt/GavenciakKK22,KnopK18,Koutecky25}, most recently by Eisenbrand and Rothvoss~\cite[Conjecture 1]{DBLP:conf/soda/EisenbrandR26}.

The existence of an FPT time algorithm for 4-block IP is arguably one of the top 3 most important open problems regarding the complexity of integer programs, together with the two problems mentioned above: the existence of a polynomial time algorithm for bounded subdeterminants, and the existence of a $2^{\O(n)}L^{\O(1)}$ algorithm.
In this paper, we provide the first FPT time algorithm for solving 4-block IPs, thereby answering the most important open question in block-structured integer programming.

\subsection*{Our contribution}
Formally, a \emph{4-block IP} is an integer program whose constraint matrix $A$ has a specific block structure consisting of $p$ global variables, $r$ global constraints, and independent small block matrices along its diagonal. In particular, we have
\begin{equation}
A =
    \label{eq:4-block-matrix}
    \begin{bmatrix}
        A_0&B_1&\cdots&B_n\\
        C_1&D_1&&\\
        \vdots&&\ddots&\\
        C_n&&&D_n
    \end{bmatrix},
\end{equation}
where $A_0\in\mathbb{Z}^{r\times p}$, $C_i\in\mathbb{Z}^{s\times p}$, $B_i\in\{-\Delta,-\Delta+1,\dots,\Delta\}^{r\times t}$, and $D_i\in\{-\Delta,-\Delta+1,\dots,\Delta\}^{s\times t}$ for $i\in[n]$.
In a sense, the class of 4-block matrices is a minimal generalization of $n$-fold and 2-stage stochastic matrices.

The new structural result to obtain the desired algorithm is \cref{lemma:2d-dilation}. We consider functions~$f\colon\Z^m\to\R$ that are \emph{integer midpoint convex}, i.e., that satisfy $f(x)\le\tfrac12f(x-p)+\tfrac12f(x+p)$ for all $x,p\in\Z^n$. We show that the restriction of such a function to a linear subspace becomes convex extensible when inspected on a sufficiently dilated lattice. A function $f\colon D\to\R\cup\{\infty\}$ is convex extensible on a set $S\subseteq D\subseteq\R^n$ if there exists a convex function $\hat f\colon\R^n\to\R\cup\{\infty\}$ so that $\hat f(x)=f(x)$ for all $x\in S$.

\begin{restatable}{lemma}{lemmatwoddilation}
	\label{lemma:2d-dilation}
	Let $f\colon\Z^m\to\R$ be integer midpoint convex, $k\in\Z_{\ge0}$, and $S\subseteq\R^m$ be a linear subspace of dimension $d\le\tfrac12k$. Then $z\mapsto f(k\cdot z)$ is convex extensible on $\Z^m\cap S$.
\end{restatable}

The proof of \cref{lemma:2d-dilation} builds upon the fact that a non-vertex integer point in a sufficiently dilated polytope can be written as the sum of two distinct integer points in the polytope, which follows from arguments by Powers and Reznick~\cite{POWERS2021106608}.

The theory developed by Ligthart~\cite{ligthart2026valuefunctionsseparableconvex} shows that value functions of block-structured integer programs become integer midpoint convex when inspected on a sufficiently dilated lattice. Tellingly, Ligthart~\cite{ligthart2026valuefunctionsseparableconvex} also mentions that in order to provide an FPT time algorithm to solve 4-block IPs, it suffices to extend this 1-dimensional midpoint convexity to convex extensibility on fixed dimensional subspaces.
This is precisely the gap that \cref{lemma:2d-dilation} closes. Therefore, we can use the algorithmic framework from~\cite{ligthart2026valuefunctionsseparableconvex}, which was inspired by~\cite{DBLP:journals/theoretics/CslovjecsekKLPP25}, to settle the complexity class of 4-block IPs. This is captured in \cref{thm:4-block-fpt}, where a function~$f\colon\R^n\to\R\cup\{\infty\}$ is separable convex if it can be written as $f(x)=\sum_{i\in[n]}f_i(x_i)$ for univariate convex functions $f_i\colon\R\to\R\cup\{\infty\}$.

\begin{restatable}{theorem}{thmfourblockfpt}
	\label{thm:4-block-fpt}
	Let $A$ be a 4-block matrix with $A_0\in\Z^{r\times p}$, $B_i\in\{-\Delta,-\Delta+1,\dots,\Delta\}^{r\times t}$, $C_i\in\Z^{s\times p}$, and $D_i\in\{-\Delta,-\Delta+1,\dots,\Delta\}^{s\times t}$ for $i\in[n]$. Let $l,u\in\Z^{p+tn}$ and let $f\colon\R^{p+tn}\to\R$ be separable convex and accessible through a comparison oracle on $\Z^{p+tn}$. Then the 4-block IP
	\[
	\min\{f(x)\ \vert\ Ax=b,\,l\le x\le u,\,x\in\Z^{p+tn}\}
	\]
	can be solved in time
	\[
	2^{\O((r+s)\Delta)^{r(r+s)}\cdot p+p^2\log p}\cdot tn\log^{\O(1)}(tn)\cdot\log^{\O(1)}\|u-l\|_\infty.
	\]
\end{restatable}

Solving the integer program entails either reporting an optimal solution or reporting that there exists no feasible solution. We measure the running time of our algorithm in terms of the number of arithmetic operations and comparison oracle queries, as is standard.

Note that it is known that the parametric dependence of the running time in \cref{thm:4-block-fpt} must be at least doubly exponential, since this is already necessary for 2-stage stochastic integer linear programs with a single global variable~\cite{DBLP:journals/mp/JansenKL23}.

We show that our algorithm can be extended to more general classes of IPs. In fact, our results are pushing the tractability results in several directions as far as possible. First, note that the objective function need not have the linear form $c^\top x$. Instead we consider separable convex objectives and it is known that it cannot be pushed beyond this class to, e.g., separable concave or general convex functions~\cite{monster}.
Second, the constraint matrix obtained by removing the global variables need not be an $n$-fold matrix, but may be a so-called tree-fold matrix or multi-stage stochastic matrix, which are recursive generalizations of the matrices in (\ref{eq:nfoldmatrix}). Note that the natural generalization of this direction would be that the matrix is, after deleting few global rows and columns, a diagonal matrix composed of $4$-block blocks; however, already this case is NP-hard~\cite{DBLP:conf/ipco/EibenGKOPW19}.
Third, the global columns may contain large coefficients, as first studied in~\cite{DBLP:journals/theoretics/CslovjecsekKLPP25}, while allowing large coefficients in the diagonal blocks leads to NP-hardness.
Another natural generalization of 4-block IPs are IPs that, after removing a parameterized number of global variables, become tractable and have bounded fractionality. We prove that our algorithmic techniques can inherently not be used to derive FPT algorithms for this class.

Finally, we expect our algorithm for 4-block IPs to find interesting applications, similarly to how the FPT algorithm for $n$-fold ILP~\cite{HemmeckeOR13} later found many use cases.
Towards this, we resolve an open question of Boehmer and Koana~\cite{DBLP:journals/talg/BoehmerK24}: the Max-Min Fair Matching problem with arbitrary size constraints is FPT parameterized by the number of recipients (\cref{sec:applications}). Its natural IP formulation is a 4-block IP that needs both global variables and global constraints.

\subsection*{Related work}
The first to study 4-block IPs were Hemmecke, Köppe, and Weismantel, who proved in~\cite{DBLP:conf/ipco/HemmeckeKW10} that uniform 4-block IPs admit polynomial time algorithms when the block dimensions are fixed (i.e., the problem is in XP). Here, a 4-block matrix is uniform if $B_1=B_2=\dots=B_n$, $C_1=C_2=\dots=C_n$, and $D_1=D_2=\dots=D_n$. Faster algorithms with slice-wise polynomial running times have been given by Chen, Kouteck{\'{y}}, Xu, and Shi~\cite{DBLP:conf/esa/0009K0S20}, Oertel, Paat, and Weismantel~\cite{DBLP:journals/mp/OertelPW24}, and Lassota and Ligthart~\cite{DBLP:conf/ipco/LassotaL26}, the latter two also allowing for non-uniform matrices.
Other special cases of which entries and structure of the blocks is allowed have been studied, see, e.g.,~\cite{DBLP:journals/disopt/ChenCZ22,DBLP:journals/mp/ChenCZ24}.

A key step towards the resolution of the 4-block problem has been the paper~\cite{DBLP:journals/theoretics/CslovjecsekKLPP25}, which studied whether large coefficients can be admitted into the constraint matrices of 2-stage and $n$-fold matrices.
The motivation for this work was precisely the 4-block problem: the authors of~\cite{DBLP:journals/theoretics/CslovjecsekKLPP25} point out that an FPT time algorithm deciding feasibility of uniform 4-block IPs would also imply an algorithm deciding feasibility of the uniform 4-block IPs where all entries in the global parts can be large (but bounded by $n$).
The key insight of~\cite{DBLP:journals/theoretics/CslovjecsekKLPP25} was in focusing on the residue classes of the right hand side vector.
This line was further developed by Eisenbrand and Rothvoss~\cite{DBLP:conf/soda/EisenbrandR26}, Lassota and Ligthart~\cite{DBLP:conf/ipco/LassotaL26}, and Ligthart~\cite{ligthart2026valuefunctionsseparableconvex} and culminates here.

In terms of lower bounds, the authors of~\cite{DBLP:conf/ipco/EibenGKOPW19} studied a generalization of 4-block IPs, namely IPs with small incidence treedepth.
They show that this class is NP-hard already for matrices with entries in $\{-1,0,1\}$ and incidence treedepth at most~5.
The incidence treedepth of 4-block IPs is bounded by their block dimensions; their solvability (despite the hardness result of~\cite{DBLP:conf/ipco/EibenGKOPW19}) stems from the fact that their \emph{incidence topological height}~\cite{monster} is 2, while the hardness construction~\cite{DBLP:conf/ipco/EibenGKOPW19} requires topological height 3.
For 4-block IP specifically, Chen, Kouteck{\'{y}}, Xu, and Shi~\cite{DBLP:conf/esa/0009K0S20} proved a $\Omega(n^r)$ lower bound on the $\ell_\infty$-norm of Graver basis elements, which are fundamental inclusion-wise minimal kernel elements, used heavily in many algorithms for IPs with bounded treedepth, such as $n$-folds and 2-stage stochastic IPs.
Hence, these frameworks could not be applied to 4-block IPs. Jansen, Klein, and Lassota~\cite{DBLP:journals/mp/JansenKL23} proved a double-exponential lower bound $2^{2^{\Omega(s+t)}}|I|^{\O(1)}$ for 2-stage stochastic IPs under the ETH, ruling out a more efficient runtime for 4-block IPs. Finally, Chen, Chen, and Zhang~\cite{DBLP:journals/disopt/ChenCZ22} showed that even $n$-fold IP, and hence 4-block IP, is NP-hard when the block dimensions are constant and the largest coefficient is part of the input, ruling out a running time polynomial in $\log\Delta$ under the assumption that $P \neq NP$.

During the preparation of this manuscript, independently and concurrently, Jansen, Ohnesorge, and Wambsganz published an FPT time algorithm for 4-block integer linear programming on arXiv using qualitatively the same framework~\cite{jansen2026fixedparameteralgorithm4blockinteger}. In comparison to them, we also establish the result for various generalizations as discussed above, including answering the open question from~\cite{jansen2026fixedparameteralgorithm4blockinteger}, of whether the techniques can be extended to separable convex objective functions.

%% file: sections/convex-extensibility.tex
\section{Convex extensibility of integer midpoint convex functions}
\label{sec:convex-extensibility}
This section is dedicated to proving the main structural result, namely, \cref{lemma:2d-dilation}. It is the key insight that allows us to apply the framework from~\cite{ligthart2026valuefunctionsseparableconvex} to obtain the desired FPT time algorithm proving~\cref{thm:4-block-fpt} in~\cref{sec:algorithm}.

The proof of \cref{lemma:2d-dilation} relies on a decomposition property of non-vertex integral points in dilated integral polytopes that was shown in~\cite{POWERS2021106608}. We spell out a variant of this with a short self-contained proof.

\begin{lemma}[Variant of Theorem 3.1 from~\cite{POWERS2021106608}]
    \label{lemma:non-vertices-are-decomposable}
    Let $P\subseteq\R^n$ be an integral polytope and $k\in\Z_{\ge0}$ be so that $P$ has affine dimension $d\le\tfrac12k$. If $x\in kP\cap\Z^n$ is not a vertex of $kP$, then there exists a nonzero $p\in\Z^n$ so that $x-p,x+p\in kP\cap\Z^n$.
\end{lemma}

\begin{proof}
    Caratheodory's theorem shows that $x$ can be written as $x=\sum_{i\in[\ell]}\lambda_iv_i$ for $\ell\le d+1$ vertices $v_i$ of $P$ with multipliers $\lambda_i\ge0$ satisfying $\sum_{i\in[\ell]}\lambda_i=k$. If there exist two distinct indices $i_1,i_2$ so that $\lambda_{i_1},\lambda_{i_2}\ge1$. Then $p=v_{i_1}-v_{i_2}\in\Z^n$ guarantees that
    \[
        x-p=\sum_{i\in[\ell]\setminus\{i_1,i_2\}}\lambda_iv_i+(\lambda_{i_1}-1)v_{i_1}+(\lambda_{i_2}+1)v_{i_2}\in kP
    \]
    since all multipliers are nonnegative. Similarly, it holds that $x+p\in kP$. Otherwise, there is only one index $i^*$ so that $\lambda_{i^*}\ge1$. In this case, $\sum_{i\in[\ell]\setminus\{i^*\}}\lambda_i\le\ell-1$, which shows that $\lambda_{i^*}\ge k-(\ell-1)$. Set $p=x-kv_{i^*}$, which is integral and nonzero because $x$ is not a vertex. We verify that $x-p=kv_{i^*}\in kP$
    and that
    \[
        x+p=2x-kv_{i^*}=\sum_{i\in[\ell]\setminus\{i^*\}}2\lambda_iv_i+(2\lambda_{i^*}-k)v_{i^*}\in kP
    \]
    since $2\lambda_{i^*}-k\ge2k-2\ell+2-k\ge k+2-2(d+1)=k-2d\ge0$.
\end{proof}

We are now ready to show \cref{lemma:2d-dilation}.

\lemmatwoddilation*

\begin{proof}
    It suffices to show that
    \begin{equation}
        \label{eq:convex-inequality}
        f(k\cdot z)\le\sum_{i\in[\ell]}\lambda_if(k\cdot z_i)
    \end{equation}
    for integer points $z_1,\dots,z_\ell,z\in\Z^m\cap S$ and nonnegative multipliers $\lambda_i$ satisfying
    \begin{equation}
        \label{eq:convex-inequality-multiplier-lp}
        \sum_{i\in[\ell]}\lambda_i=1\text{ and }\sum_{i\in[\ell]}\lambda_iz_i=z.
    \end{equation}
    For the sake of showing (\ref{eq:convex-inequality}), we may assume that $\lambda$ minimizes the right-hand side of (\ref{eq:convex-inequality}) under the constraints (\ref{eq:convex-inequality-multiplier-lp}). In particular, we may assume that it is a vertex solution of the corresponding linear program, of which the nonzero basic coordinates have corresponding constraint matrix columns $(z_i,1)$ that are linearly independent. This allows us to assume that all $z_i$ in (\ref{eq:convex-inequality}) are affinely independent and guarantees that there exists an affine function $h\colon\R^n\to\R$ so that $h(k\cdot z_i)=f(k\cdot z_i)$ and $h(k\cdot z)=\sum_{i\in[\ell]}\lambda_if(k\cdot z_i)$. Consider $g\colon\Z^n\to\R$ defined by $g(x)=f(x)-h(x)$, which is integer midpoint convex. Define the simplex $P=\conv\{z_1,\dots,z_\ell\}$, which is of affine dimension at most $d$. Let $x^*$ be a vertex of the set $X^*=\conv\{x\in kP\cap\Z^m:g(x)=g^*\}$ where $g^*=\max_{x\in kP\cap\Z^m}g(x)$. Suppose for a contradiction that $x^*$ is not a vertex of $kP$, then \cref{lemma:non-vertices-are-decomposable} shows that there exists a nonzero $p\in\Z^m$ so that $x-p,x+p\in kP\cap\Z^m$. Optimality of $x^*$ shows that $g(x-p),g(x+p)\le g(x^*)$. On the other hand, integer midpoint convexity shows that
    \[
        g(x^*)\le\tfrac12g(x-p)+\tfrac12g(x+p)\le\tfrac12g(x^*)+\tfrac12g(x^*),
    \]
    showing that $g(x^*)=g(x-p)=g(x+p)$, i.e., $x-p,x+p\in X^*$. This contradicts that $x^*$ was a vertex of $X^*$. Hence, $x^*=k\cdot z_{i^*}$ for some $i^*\in[\ell]$, which shows that
    \[
        f(k\cdot z)-h(k\cdot z)= g(k\cdot z)\le g(x^*)=g(k\cdot z_{i^*})=f(k\cdot z_{i^*})-f(k\cdot z_{i^*})=0,
    \]
    which shows the desired inequality.
\end{proof}

%% file: sections/algorithm.tex
\section{Algorithm}
\label{sec:algorithm}
We first spell out the necessary definitions and results from~\cite{ligthart2026valuefunctionsseparableconvex} so that we can apply~\cref{lemma:2d-dilation} to derive the desired FPT time algorithm.

In~\cite{ligthart2026valuefunctionsseparableconvex}, Ligthart connects the \emph{integer decomposition property} of polyhedra with convex extensibility of value functions of integer programs. The value function of an integer program is the function that maps the vector $b$ of right-hand sides to the optimal objective value of an optimal solution. We consider value functions defined on $\Z^m$ of the form
\[
    b\mapsto\min\{f(x)\ \vert\ Ax=b,\ x\in\Z^n\},
\]
for constraint matrices $A\in\Z^{m\times n}$ and separable convex objective functions $f\colon\R^n\to\R\cup\{\infty\}$. Note that infinite values of $f$ are used to model variable lower and upper bounds on $x$. The value function itself is $\infty$ if it admits no solution with finite objective value and $-\infty$ when the problem is unbounded.

A polyhedron $P\subseteq\R^n$ has the integer decomposition property (IDP) if for all $k\in\Z_{\ge0}$ and $x\in kP\cap\Z^n$ there exist $x_1,\dots,x_k\in P\cap\Z^n$ so that $x=x_1+\dots+x_k$. In~\cite{ligthart2026valuefunctionsseparableconvex}, it is said that a class of constraint matrices $\mathcal A$ has the IDP after an $M$-dilation if the dilated polyhedron $M\cdot\{x\in\R_{\ge0}^n:Ax=b\}$ has the IDP for all integral vectors $b\in\Z^m$ and matrices $A\in\Z^{m\times n}$ in $\mathcal A$. For the purpose of proving \cref{thm:4-block-fpt}, we focus on the class $\nfold(r,s,\Delta)$, which consists of all matrices of the form $A_{\text{n-fold}}$ in~(\ref{eq:nfoldmatrix}) with $B_i\in\{-\Delta,-\Delta+1,\dots,\Delta\}^{r\times t_i}$ and $D_i\in\{-\Delta,-\Delta+1,\dots,\Delta\}^{s\times t_i}$ for $i\in[n]$.

In this setting, Ligthart~\cite{ligthart2026valuefunctionsseparableconvex} shows that the dilation $M$ that is needed to establish the IDP is connected to the convex extensibility of IP value functions on intersections of dilated lattice translates and lines:

\begin{lemma}[Proposition 7 in~\cite{ligthart2026valuefunctionsseparableconvex}]
    \label{lemma:idp-iff-convex-extensibility-along-line}
    Let $\mathcal A$ be a class of constraint matrices that is closed under inverting columns. Then $\mathcal A$ has the IDP after an $M$-dilation if and only if the value function $h\colon\Z^m\to\R\cup\{-\infty,\infty\}$ given by $h(z)=\min\{f(x)\ \vert\ Ax=r+Mz,\,x\in\Z^n\}$ is convex extensible on $\Z^m\cap L$ for any separable convex $f\colon\R^n\to\R\cup\{\infty\}$, offset vector $r\in\Z^m$, and line $L\subseteq\R^m$.
\end{lemma}
Here, $\mathcal A$ is closed under inverting columns if inverting the sign of all coefficients in a column of a matrix in $\mathcal A$ results in a matrix that is also in $\mathcal A$, which is satisfied by $\nfold(r,s,\Delta)$. Ligthart~\cite{ligthart2026valuefunctionsseparableconvex} shows that $\nfold(r,s,\Delta)$ (along with other block-structured constraint matrices) are well-behaved by showing that it has the IDP after a dilation that is bounded by a function of the parameters $r,s,\Delta$, independently of the number of blocks $n$:

\begin{lemma}[Corollary 12 in~\cite{ligthart2026valuefunctionsseparableconvex}]
    \label{lemma:n-fold-scaled-idp}
    Let $r,s,\Delta\in\Z_{\ge0}$ be given. Then there exists a positive integer $M=2^{\O((r+s)\Delta)^{r(r+s)}}$ so that $\nfold(r,s,\Delta)$ has the IDP after an $M$-dilation.
\end{lemma}

Together, \cref{lemma:idp-iff-convex-extensibility-along-line,lemma:n-fold-scaled-idp} show that if $A\in\nfold(r,s,\Delta)$, then the function
\[
    h(z)=\min\bigl\{f(x)\bigm\vert Ax=r+Mz,\,x\in\Z^n\bigr\}
\]
is convex extensible on integer lines for a parameterized value of $M$. This is exploited in~\cite{ligthart2026valuefunctionsseparableconvex} to give an algorithm that solves 4-block IPs with exactly one global variable. The (more general) algorithmic framework, also employed to tackle other parameterized integer programming problems in~\cite{ligthart2026valuefunctionsseparableconvex}, is formulated in \cref{lemma:fixed-phase-value-function-reformulation}. It is based on value function reformulations of (\ref{eq:generic-2-stage-ip}) and a remainder guessing strategy originating from~\cite{DBLP:journals/theoretics/CslovjecsekKLPP25}:

\begin{lemma}[Lemma 18 in~\cite{ligthart2026valuefunctionsseparableconvex} (for $d=m$)]
    \label{lemma:fixed-phase-value-function-reformulation}
    Let $A\in\Z^{m\times n}$, $p,U\in\Z_{\ge0}$, and $M\in\Z_{\ge1}$ be such that:
    \begin{itemize}
        \item the value function $h\colon\Z^m\to\R\cup\{-\infty,\infty\}$ given by
        \begin{equation}
            \label{eq:value-function-assumed-convex-extensibility}
            h(z)=\min\bigl\{\overline f(\overline x)\bigm\vert[A,I]\overline x=r+Mz,\,\overline x\in\Z^{n+m}\bigr\}
        \end{equation}
        is convex extensible on $\Z^m\cap S$ for any $r\in\Z^m$, linear subspace $S$ with dimension at most $p$, and separable convex objective function $\overline f\colon\R^{n+m}\to\R\cup\{\infty\}$,
        \item an optimal solution to
        \begin{equation}
            \label{eq:generic-second-stage-ip}
            \min\bigl\{\overline f(\overline x)\bigm\vert[A,I]\overline x=b,\,\overline l\le\overline x\le\overline u,\,\overline x\in\Z^{n+m}\bigr\},
        \end{equation}
        can be found in time $T_{\mathrm{opt}}(\sigma)$ for any $\overline l,\overline u\in\Z^{n+m}$, $b\in\Z^m$, and separable convex $\overline f\colon\R^{n+m}\to\R$ assuming that: an initial feasible solution is provided, $\|\overline u-\overline l\|_\infty\le\sigma$, and $\overline f$ is accessible through a comparison oracle on $\Z^{n+m}$,
        \item the product $Ax\in\Z^m$ can be computed in time $T_{\mathrm{mul}}$ for any $x\in\Z^n$,
        \item the induced matrix norm $\|A\|_\infty$ is bounded by $U$ from above.
    \end{itemize}
    Let $C\in\Z^{m\times p}$, $l,u\in\Z^n$, $b\in\Z^m$, $V\in\R^{\ell\times p}$, $w\in\R^\ell$, $c\in\R^p$, and $\rho>0$ be such that $\{y\in\R^p:Vy\le w\}\subseteq\{y\in\R^p:\|y-c\|_\infty\le\rho\}$. Let $g\colon\R^p\to\R$ be convex and $f\colon\R^n\to\R$ be separable convex and let both be accessible through a comparison oracle on $\Z^{p+n}$. Then
    \begin{equation}
        \label{eq:generic-2-stage-ip}
        \min\bigl\{g(y)+f(x)\bigm\vert Cy+Ax=b,\,Vy\le w,\,y\in\Z^p,\,l\le x\le u,\,x\in\Z^n\bigr\}
    \end{equation}
    can be solved in time
    \[
        M^p\cdot2^{\O(p^2\log p)}\cdot\Bigl(\log^{\O(1)}(\rho)+\log(\rho)\cdot\bigl(T_{\mathrm{opt}}((1+2U)\|u-l\|_\infty)+T_{\mathrm{mul}}+m+\ell\bigr)\Bigr).
    \]
\end{lemma}
Here, $[A,I]$ denotes the horizontal concatenation of $A$ and the identity matrix $I$.

We provide some background on the algorithm behind \cref{lemma:fixed-phase-value-function-reformulation}. In~\cite{ligthart2026valuefunctionsseparableconvex}, (\ref{eq:generic-2-stage-ip}) is solved by minimizing an objective function of the form $g(y)+h'(b-Ay)$ exclusively over the variables~$y\in\Z^p$, by abstracting away the objective contribution of the variables $x$ and moving it to the term $h'(b')=\min\{f(x)\ \vert\ Ax=b',\,l\le x\le u,\,x\in\Z^n\}$. At first, the objective on $y$ is not convex extensible, meaning it is not (yet) able to be minimized with fixed-dimension convex integer programming techniques. However, the assumption on (\ref{eq:value-function-assumed-convex-extensibility}) guarantees that $h'$ is convex extensible on a set of the form $(r'+M\Z^m)\cap S'$ for some fixed $r'\in\Z^m$ and affine subspace $S'$ with dimension at most $p$. To exploit this, one guesses the residue vector $r\in\{0,1,\dots,M-1\}^p$ for which there exists an optimal solution $y$ that is of the form $r+Mv$ for some $v\in\Z^p$. Under such restriction, the argument of $h'$ lies exactly on a set that guarantees convex extensibility. This makes it possible to apply the fixed-dimension algorithm from~\cite{DBLP:journals/dam/VeselovGZC20} to find the optimal $v$ for this problem, which can deal with the convex extensible objective function as long as it is finite.

To guarantee the finiteness of the objective function, Ligthart~\cite{ligthart2026valuefunctionsseparableconvex} introduces dummy variables $s$ to incur a high, but finite, cost when a subproblem $Ax=b',\,l\le x\le u$ is infeasible. For this reason, the original constraint matrices $A$ in (\ref{eq:value-function-assumed-convex-extensibility}) and (\ref{eq:generic-second-stage-ip}) are augmented with an identity matrix. The particular choice of $\overline f$ for (\ref{eq:value-function-assumed-convex-extensibility}) in~\cite{ligthart2026valuefunctionsseparableconvex} is given by $\overline f(x,s)=f(x)+\delta(s)+\xi(x)$, where $f\colon\R^n\to\R$ is the original finite objective function from (\ref{eq:generic-2-stage-ip}), $\delta(s)=\alpha\cdot\|s\|_1$ for $\alpha>0$, and $\xi(x)$ is the indicator function that is $0$ if $l\le x\le u$ and $\infty$ otherwise. The combination of $\delta\ge0$ and $\xi$, which bounds the effective domain of $x$, guarantees that $h>-\infty$. Additionally, since we may assume that $l\le u$, the system $[A,I]\overline x=b$ always has one solution $(l,b-Al)$ with finite objective value, which guarantees that $h<\infty$. Hence, the proof in~\cite{ligthart2026valuefunctionsseparableconvex} only uses the convex extensibility of (\ref{eq:value-function-assumed-convex-extensibility}) on $\Z^m\cap S$ for choices of $\overline f$ that guarantee that the value function
\begin{equation}
    \label{eq:value-function-with-slack}
    b\mapsto\min\bigl\{\overline f(\overline x)\bigm\vert[A,I]\overline x=b,\,\overline x\in\Z^{n+m}\bigr\}
\end{equation}
is finite on the entirety of $\Z^m$.

It is noted in~\cite{ligthart2026valuefunctionsseparableconvex} that if one can meet the convex extensibility condition for (\ref{eq:value-function-assumed-convex-extensibility}) for $A\in\nfold(r,s,\Delta)$ for any fixed dimension $p$, this would extend the 4-block algorithm from one global variable to any parameterized number of global variables. Since convex extensibility along $\Z^n\cap L$ for any line $L\subseteq\R^n$ implies integer midpoint convexity and we can assume that (\ref{eq:value-function-with-slack}) is finite, we can now close this gap by using \cref{lemma:2d-dilation}.

\begin{lemma}
    \label{lemma:strengthened-fixed-phase-value-function-reformulation}
    \cref{lemma:fixed-phase-value-function-reformulation} holds even if we only demand the midpoint integer convexity of (\ref{eq:value-function-assumed-convex-extensibility}) for all separable convex $\overline f\colon\R^{n+m}\to\R\cup\{\infty\}$.
\end{lemma}

\begin{proof}
    \cref{lemma:2d-dilation} directly shows that $h'\colon\Z^d\to\R$ given by
    \[
        h'(z)=\min\bigl\{\overline f(\overline x)\bigm\vert[A,I]\overline x=r+(2pM)\cdot z,\,\overline x\in\Z^{n+m}\bigr\}
    \]
    is convex extensible on $\Z^m\cap S$ for any $r\in\Z^m$, linear subspace $S$ with dimension at most $p$, and separable convex objective function $\overline f\colon\R^{n+m}\to\R\cup\{\infty\}$ so that (\ref{eq:value-function-with-slack}) is finite. Hence, we may apply \cref{lemma:fixed-phase-value-function-reformulation} for $M'=2pM$ to solve (\ref{eq:generic-2-stage-ip}). Note that the running time factor that depends on $p$ remains
    \[
        (2pM)^p\cdot2^{\O(p^2\log p)}=M^p\cdot2^{\O(p\log p)}\cdot2^{\O(p^2\log p)}=M^p\cdot2^{\O(p^2\log p)}.\qedhere
    \]
\end{proof}

Now \cref{thm:4-block-fpt} follows from a straightforward application of \cref{lemma:strengthened-fixed-phase-value-function-reformulation}.

\thmfourblockfpt*

\begin{proof}
    Write $l=(l_0,\hat l)$ and $u=(u_0,\hat u)$ for $l_0,u_0\in\Z^p$ and $\hat l,\hat u\in\Z^{tn}$. We apply \cref{lemma:strengthened-fixed-phase-value-function-reformulation} with the constraint matrices
    \[
        C=\begin{bmatrix}
            A_0\\
            C_1\\
            \vdots\\
            C_n
        \end{bmatrix},\text{ and }A=\begin{bmatrix}
            B_1&\cdots&B_n\\
            D_1&&\\
            &\ddots&\\
            &&D_n
        \end{bmatrix}.
    \]
    We encode the global variable bounds $l_0\le y\le u_0$ in $Vy\le w$ with $\ell=2p$ constraints. We directly insert the right-hand side vector $b$, the $n$-fold variable bounds $\hat l,\hat u$, and the objective function $(y,x)\mapsto f(y,x)$ into \cref{lemma:strengthened-fixed-phase-value-function-reformulation}. We can take $\rho\le\|u_0-l_0\|_\infty\le\|u-l\|_\infty$. A product $Ax$ can be computed in time $T_{\mathrm{mul}}=\O((r+s)tn)$ and $\|A\|_\infty\le\Delta tn=:U$. Note that $[A,I]\in\nfold(r,s,\Delta)$ (after permuting columns). Therefore, we can solve (\ref{eq:generic-second-stage-ip}) using the dual treedepth algorithm from~\cite{DBLP:conf/ipco/HunkenschroderKLV25}, which runs in time $T_{\mathrm{opt}}(\sigma)=(rs\Delta)^{\O(rs(r+s))}\cdot(tn+r+sn)\log(tn+r+sn)\cdot\log\sigma=(rs\Delta)^{\O(rs(r+s))}\cdot tn\log(tn)\cdot\log\sigma$. \cref{lemma:n-fold-scaled-idp} shows that $[A,I]$ has the IDP after an $M$-dilation where $M=2^{\O((r+s)\Delta)^{r(r+s)}}$. Ligthart~\cite{ligthart2026valuefunctionsseparableconvex} notes that the value of $M$ can be computed in $\O(M)$ time. \cref{lemma:idp-iff-convex-extensibility-along-line} shows that (\ref{eq:value-function-assumed-convex-extensibility}) is convex extensible on any line, which, in particular, implies that it is integer midpoint convex. Now \cref{lemma:strengthened-fixed-phase-value-function-reformulation} shows that we can solve the 4-block IP in time
    \begin{alignat*}{2}
        &M^p\cdot{}&&2^{\O(p^2\log p)}\cdot\Bigl(\log^{\O(1)}(\rho)+\log(\rho)\\
        &&&\cdot\bigl(T_{\mathrm{opt}}((1+2U)\|u-l\|_\infty)+T_{\mathrm{mul}}+(r+sn)+\ell\bigr)\Bigr)\\
        ={}&\mathrlap{\bigl(2^{\O((r+s)\Delta)^{r(r+s)}}\bigr)^p\cdot2^{\O(p^2\log p)}\cdot\Bigl(\log^{\O(1)}(\|u-l\|_\infty)}&&\\
        &&&\cdot\bigl(T_{\mathrm{opt}}((1+2\Delta tn)\|u-l\|_\infty)+\O((r+s)tn)+(r+sn)+2p\bigr)\Bigr)\\
        ={}&\mathrlap{2^{\O((r+s)\Delta)^{r(r+s)}\cdot p+p^2\log p}\cdot\Bigl(\log^{\O(1)}(\|u-l\|_\infty)\cdot\bigl((rs\Delta)^{\O(rs(r+s))}\cdot tn\log(tn)}&&\\
        &&&\cdot\log((1+2\Delta tn)\|u-l\|_\infty)+\O((r+s)tn)+(r+sn)+2p\bigr)\Bigr)\\
        ={}&\mathrlap{2^{\O((r+s)\Delta)^{r(r+s)}\cdot p+p^2\log p}\cdot\Bigl(\log^{\O(1)}(\|u-l\|_\infty)\cdot\bigl(tn\log(tn)}&&\\
        &&&\cdot\bigl(\log(\Delta)+\log(tn)+\log(\|u-l\|_\infty)\bigr)+\O((r+s)tn)+(r+sn)+2p\bigr)\Bigr)\\
        ={}&\mathrlap{2^{\O((r+s)\Delta)^{r(r+s)}\cdot p+p^2\log p}\cdot tn\log^{\O(1)}(tn)\cdot\log^{\O(1)}\|u-l\|_\infty.}&&\tag*{\qedhere}
    \end{alignat*}
\end{proof}

\subsection{Other implications}

We now list a number of alternative algorithms for solving 4-block IPs, which can be derived in a similar way to \cref{thm:4-block-fpt} using~\cref{lemma:strengthened-fixed-phase-value-function-reformulation} and combining it with known proximity arguments or rewriting techniques. The proofs are postponed to \cref{appendix:omitted-proofs}.

If the objective function $f$ is a linear function $f(x)=c^\top x$ with explicitly given coefficients~$c\in\R^n$, and if the absolute value of all constraint matrix coefficients are bounded by $\Delta$, then a 4-block IP can be solved in strongly FPT time, with a number of arithmetic operations that is independent of $\|u-l\|_\infty$. To accomplish this, we first solve the continuous relaxation using the strongly polynomial linear programming algorithm by Tardos~\cite{DBLP:journals/ior/Tardos86} and then employ a proximity bound such as the one from~\cite{DBLP:journals/mp/CookGST86}.

\begin{restatable}{corollary}{corollarystronglyfptfourblock}
    \label{corollary:strongly-fpt-4-block}
    Let $A$ be a 4-block matrix with $A_0\in\{-\Delta,-\Delta+1,\dots,\Delta\}^{r\times p}$, $B_i\in\{-\Delta,-\Delta+1,\dots,\Delta\}^{r\times t}$, $C_i\in\{-\Delta,-\Delta+1,\dots,\Delta\}^{s\times p}$, and $D_i\in\{-\Delta,-\Delta+1,\dots,\Delta\}^{s\times t}$ for $i\in[n]$. Let $l\in(\Z\cup\{-\infty\})^{p+tn}$, $u\in(\Z\cup\{\infty\})^{p+tn}$ and let $c\in\R^{p+tn}$. Then
    \begin{equation}
        \label{eq:4-block-ilp}
        \min\{c^\top x\ \vert\ Ax=b,\,l\le x\le u,\,x\in\Z^{p+tn}\}
    \end{equation}
    can be solved in time $g(p,r,s,\Delta)\cdot(tn)^{\O(1)}$, where $g$ is a computable function.
\end{restatable}
Here, infinite lower or upper bounds correspond to the absence of such bounds, and solving the IP also entails detecting whether the integer program has an unbounded objective value. The stated running time applies under the assumption that rounding can be done with a single arithmetic operation, which is standard in the literature.

It is also possible to obtain FPT algorithms for 4-block IPs when large coefficients are present in the blocks $B_1,\dots,B_n$ under either i) an uniformity assumption: $B_1=\dots=B_n$ and $t$ being a parameter, or ii) the assumption that the coefficients in $B_i\in K^{r\times t}$ come from a set $K\subseteq\Z$ of parameterized cardinality. This follows from the fact that such a 4-block IP can be rewritten to an equivalent 4-block with small coefficients in $B_i$, which is illustrated in Section 3.2 of~\cite{ligthart2026valuefunctionsseparableconvex}, using (a variant of) the rewriting technique by Chen, Chen, and Zhang~\cite{DBLP:journals/disopt/ChenCZ22}. Note that the additional assumptions are necessary in order to not capture the NP-hard subset sum problem.

\begin{restatable}{proposition}{propositionuniformlargecoefficientfourblock}
    \label{proposition:uniform-large-coefficient-4-block-fpt}
    Let $A$ be a 4-block matrix with $A_0\in\Z^{r\times p}$, $B_i=B$, $C_i\in\Z^{s\times p}$, and $D_i\in\{-\Delta,-\Delta+1,\dots,\Delta\}^{s\times t}$ for $i\in[n]$ and a uniform $B\in\Z^{r\times t}$. Let $l,u\in\Z^{p+tn}$ and let $f\colon\R^{p+tn}\to\R$ be separable convex and accessible through a comparison oracle on $\Z^{p+tn}$. Then
    \[
        \min\{f(x)\ \vert\ Ax=b,\,l\le x\le u,\,x\in\Z^{p+tn}\}
    \]
    can be solved in time $g(p,t,s,\Delta)\cdot(r+n)\log^{\O(1)}(n)\cdot\log^{\O(1)}\|u-l\|_\infty$, where $g$ is a computable function.
\end{restatable}

\begin{restatable}{proposition}{propositionnonuniformlargecoefficientfourblock}
    \label{proposition:non-uniform-large-coefficient-4-block-fpt}
    Let $\{k_1,\dots,k_\ell\}=K\subseteq\Z$. Let $A$ be a 4-block matrix with $A_0\in\Z^{r\times p}$, $B_i\in K^{r\times t}$, $C_i\in\Z^{s\times p}$, and $D_i\in\{-\Delta,-\Delta+1,\dots,\Delta\}^{s\times t}$ for $i\in[n]$. Let $l,u\in\Z^{p+tn}$ and let $f\colon\R^{p+tn}\to\R$ be separable convex and accessible through a comparison oracle on $\Z^{p+tn}$. Then
    \[
        \min\{f(x)\ \vert\ Ax=b,\,l\le x\le u,\,x\in\Z^{p+tn}\}
    \]
    can be solved in time $g(p,r,s,\Delta,\ell)\cdot tn\log^{\O(1)}(tn)\cdot\log^{\O(1)}\|u-l\|_\infty$, where $g$ is a computable function.
\end{restatable}

\subsection{Variable backdoors in integer programs}

The class of 4-block integer programs can be viewed in the light of \emph{(deletion) backdoors to triviality} in the sense that, after removing a small set of variables, a 4-block IP becomes an integer program that belongs to the tractable class of $n$-fold IPs. The set of variables to be removed is in this case called the \emph{(variable) backdoor} and its size can be interpreted as the ``distance'' to a ``trivial'' IP instance as discussed in~\cite{DBLP:books/ox/Niedermeier06}.

In this light, \cref{lemma:strengthened-fixed-phase-value-function-reformulation} provides a more general tool for solving an IP in FPT time if it has a parameter-sized variable backdoor to a tractable IP. Phrased alternatively, it applies to IPs with constraint matrices of the form $[C,A]$, with arbitrary $C$, but well-structured $A$. As a result of the characterization of \cref{lemma:idp-iff-convex-extensibility-along-line}, the tool from \cref{lemma:strengthened-fixed-phase-value-function-reformulation} can be applied when:
\begin{itemize}
    \item $[A,I]$ belongs to a class of matrices that is closed under inverting columns and has the IDP after an $M$-dilation, with $M$ being a function of the relevant parameters; and
    \item an IP with constraint matrix $[A,I]$ is tractable (e.g., solvable in FPT time);
\end{itemize}

The class of $n$-fold IPs are not the only class of constraint matrices that satisfy these properties: recursive generalizations of 2-stage stochastic and $n$-fold IPs, called multistage and tree-fold IPs respectively, or equivalently, matrices of bounded primal treedepth or bounded dual treedepth (with bounded coefficients) also satisfy both of these conditions. This is witnessed by dimension independent dilation bounds that establish the IDP given in~\cite{ligthart2026valuefunctionsseparableconvex} and the FPT algorithms from~\cite{DBLP:conf/ipco/HunkenschroderKLV25} that minimize separable convex objective functions over such constraint matrices. We refer to~\cite{DBLP:journals/mor/EisenbrandHKKLO25,DBLP:conf/ipco/HunkenschroderKLV25} for a formal definition of the treedepth parameters. This consideration results in \cref{proposition:variables-plus-treedepth}, of which the proof is postponed to \cref{appendix:omitted-proofs}.

\begin{restatable}{proposition}{propositionvariablesplustreedepth}
    \label{proposition:variables-plus-treedepth}
    Let $C\in\Z^{m\times p}$. Let $A\in\{-\Delta,-\Delta+1,\dots,\Delta\}^{m\times n}$ have primal or dual treedepth at most $d$. Let $l,u\in\Z^{p+n}$ and let $f\colon\R^{p+n}\to\R$ be separable convex and accessible through a comparison oracle on $\Z^{p+n}$. Then
    \[
        \min\{f(x)\ \vert\ [C,A]x=b,\,l\le x\le u,\,x\in\Z^{p+n}\}
    \]
    can be solved in time $g(p,d,\Delta)\cdot(n+m)\log^{\O(1)}(n+m)\cdot\log^{\O(1)}\|u-l\|_\infty$, where $g$ is a computable function.
\end{restatable}

The FPT time algorithm from~\cite{lassota2025parameterizedalgorithmsmatchinginteger} also solves IPs with constraint matrices of the form $[C,A]$, but requires that $A\in\Z^{m\times n}$ has columns with an $\ell_1$-norm of at most $2$, i.e., $A$ corresponds to a matching problem in a bidirected graph. This example shows that block-structured IPs are not the only class of IPs that can be extended with a parameterized number of complicating variables while retaining tractability. Now, one may be inclined to ask whether IPs with constraint matrix of the form $[C,A]$ are always FPT parameterized by the number of columns $p$ of $C$ if IPs with constraint matrix $A$ are tractable. However, the authors of~\cite{DBLP:journals/ai/DvorakEGKO21} show that this is not the case: the class of block-diagonal IPs with blocks of bounded size and unary encoded coefficients is solvable in polynomial time by applying Lenstra's algorithm to independently solve each subproblem, whereas adding only 3 variables to this block-diagonal matrix to obtain a 2-stage stochastic IP with unary encoded coefficients results in an NP-hard class of IPs. It also deserves to be noted that not every tractable extension of the form $[C,A]$, is derivable via the mechanism of \cref{lemma:strengthened-fixed-phase-value-function-reformulation}: an obvious example is the fact that any constraint matrix with a fixed number of columns can be split arbitrarily into the matrices $C$ and $A$, whereas the dilation $M$ required to obtain convex extensibility for the IP value function associated with $A$ depends exponentially on the encoding length of the coefficients. On the other hand, IPs with a fixed number of variables and binary encoded coefficients can be solved using Lenstra's algorithm, ignoring the arbitrary variable split.

A natural refinement of the previous question would be whether \cref{lemma:strengthened-fixed-phase-value-function-reformulation} can be applied to extend any tractable class of IPs which can ``inherently not capture large coefficients'', with a fixed number of additional variables. A natural requirement on such class would be to require that the polyhedron
\begin{equation}
    \label{eq:parametric-slack-polyhedron}
    \{x\in\R_{\ge0}^{n+m}:[A,I]x=b\}
\end{equation}
is $\tfrac1q$-integral for all $b\in\Z^m$ and some bounded $q\in\Z_{\ge1}$, i.e., all vertices of (\ref{eq:parametric-slack-polyhedron}) belong to $\tfrac1q\Z^{n+m}$. As shown in~\cite[Corollary 3.9]{ekbatani2022circuit}, this is equivalent to $[A,I]$ having circuits of bounded $\ell_\infty$-norm. This is motivated by the fact that IPs with bounded treedepth~\cite{DBLP:conf/aaai/BrandKO21} and the matching IPs considered in~\cite{lassota2025parameterizedalgorithmsmatchinginteger} are known to have such bounded fractionality. Additionally, an important special case in this setting occurs when (\ref{eq:parametric-slack-polyhedron}) is integral for all $b\in\Z^m$, which is equivalent to $A$ being totally unimodular by the Hoffman-Kruskal theorem. Baum and Trotter~\cite{10.1007/978-3-642-95322-4_2} have shown that this is equivalent to the, seemingly stronger, property that (\ref{eq:parametric-slack-polyhedron}) has the IDP for all $b\in\Z^m$.

However, already for $\tfrac12$-integral polyhedra, \cref{proposition:bounded-fractionality-is-insufficient} shows that the story is different. In particular, it shows that having the IDP after a bounded dilation is a strictly stronger assumption than having bounded fractionality for all $b\in\Z^m$.

\begin{restatable}{proposition}{propositionboundedfractionalityisinsufficient}
    \label{proposition:bounded-fractionality-is-insufficient}
    For all positive integers $T$ there exists a constraint matrix $A\in\{-1,0,1\}^{m\times n}$, $b^*\in\Z^m$, and $r^*\in\Z$ so that, if $P_b=\{y\in\R_{\ge0}^{n+m}:[A,I]y=b\}$, then
    \begin{itemize}
        \item $P_b$ is $\tfrac12$-integral for all $b\in\Z^m$;
        \item if $M\in[T]$, then $M\cdot P_{b^*}$ does not have the IDP; and
        \item if $M\in[T]$, then function $h\colon\Z\to\R$ given by $h(t)=\min\{\|s\|_1\ \vert\ Ax+s=(0,\dots,0,t),\,x\in\{0,1\}^n,\,s\in\Z^m\}$ satisfies $h(r^*)>\tfrac12h(r^*-M)+\tfrac12h(r^*+M)$;
    \end{itemize}
\end{restatable}

The proof of \cref{proposition:bounded-fractionality-is-insufficient} is postponed to \cref{appendix:omitted-proofs}. We note that bounded fractionality itself does not necessarily imply tractability of the corresponding IPs: the stable set polytope is $\tfrac12$-integral, but optimization over it is NP-hard. In addition, \cref{proposition:bounded-fractionality-is-insufficient} does not exlude the possibility of an algorithmic mechanism that solves any IP with bounded fractionality extended with a parameterized number of integer variables at the cost of an FPT time overhead, but merely shows that this cannot be achieved by \cref{lemma:strengthened-fixed-phase-value-function-reformulation}.

%% file: sections/appendix.tex
\section{Omitted proofs}
\label{appendix:omitted-proofs}

In this section, we provide the proofs of \cref{corollary:strongly-fpt-4-block,proposition:uniform-large-coefficient-4-block-fpt,proposition:non-uniform-large-coefficient-4-block-fpt,proposition:variables-plus-treedepth,proposition:bounded-fractionality-is-insufficient}. Notationwise, we use $\veczero$ and $\vecone$ to respectively denote the all-zero and all-ones vector of appropriate dimension.

\corollarystronglyfptfourblock*

\begin{proof}
    Apply Tardos' algorithm~\cite{DBLP:journals/ior/Tardos86} to solve the LP relaxation
    \begin{equation}
        \label{eq:4-block-lp}
        \min\{c^\top y\ \vert\ Ay=b,\,l\le y\le u,\,y\in\R^{p+tn}\}
    \end{equation}
    in polynomial time with a number of arithmetic operations that depends only on $A$, i.e., depending only on $p+tn$, $r+sn$, and $\Delta$. If the LP relaxation is infeasible, then so is the original IP. If the LP relaxation is unbounded, the IP is unbounded if and only if it is feasible, which can be determined by solving (\ref{eq:4-block-ilp}) for $c=0$, which will have a bounded LP relaxation. Hence, we may assume that we are given an optimal solution $y^*$ to (\ref{eq:4-block-lp}). The proximity bound from~\cite[Theorem 1]{DBLP:journals/mp/CookGST86} (see also~\cite[Theorem 3.14]{HemmeckeKW14} for the application to equality-constrained IPs) implies that there exists an optimal solution $x^*$ to the 4-block IP satisfying $\|x^*-y^*\|_\infty\le B$, where $B$ equals the number of variables $p+tn$ times the maximum absolute value of a subdeterminant of $A$. The Hadamard bound shows that $B\le(p+tn)\cdot(\sqrt{p+tn}\cdot\Delta)^{p+tn}$. By updating the variable bounds to $l_j'=\max\{l_j,\lceil y_j^*-B\rceil\}$ and $u_j'=\min\{u_j,\lfloor y_j^*+B\rfloor\}$ for $j\in[p+tn]$, we do not lose the optimal solution $x^*$ in (\ref{eq:4-block-ilp}). Therefore, we can solve (\ref{eq:4-block-ilp}) with these new variable bounds using \cref{thm:4-block-fpt} in time
    \[
        g(p,r,s,\Delta)\cdot(tn)^{\O(1)}\cdot\log^{\O(1)}\|u'-l'\|_\infty,
    \]
    where $\log\|u'-l'\|_\infty\le\log(1+2B)=(p+tn)^{\O(1)}\cdot\log\Delta$. This results in the claimed running time.
\end{proof}

\propositionuniformlargecoefficientfourblock*

\begin{proof}
    Write $l=(l_0,\hat l)$, $u=(u_0,\hat u)$ for $l_0,u_0\in\Z^p$, $\hat l,\hat u\in\Z^{tn}$, and $x=(x_0,x_1,\dots,x_n)$ for $x_0\in\Z^p,x_1\in\Z^t,\dots,x_n\in\Z^t$, and $b=(b_0,\hat b)$ for $b_0\in\Z^r$, $\hat b\in\Z^{sn}$.

    We apply \cref{lemma:strengthened-fixed-phase-value-function-reformulation} with the constraint matrices
    \[
        C=\left[\begin{array}{c c}
            &-I\\
            C_1&\\
            \vdots\\
            C_n&
        \end{array}\right],\text{ and }A=\begin{bmatrix}
            I&\cdots&I\\
            D_1&&\\
            &\ddots&\\
            &&D_n
        \end{bmatrix}.
    \]
    We use the right-hand side vector $b'=(\veczero,\hat b)$. By writing $y=(x_0,z)$ for $z\in\Z^t$, it becomes clear that the top $t$ rows encode $z=\sum_{i\in[n]}x_i$. The remaining rows of $C$ and $A$ encode the local constraints of the original 4-block instance. The uniformity of the global constraints allows us to encode $A_0x_0+Bx_1+\dots+Bx_n=A_0x_0+Bz=b_0$ directly in $Vy\le w$. In addition to these $2r$ constraints, we add the global variable bounds $l_0\le x_0\le u_0$ and admissible aggregated variable bounds $\sum_{i\in[n]}l_i\le z\le\sum_{i\in[n]}u_i$ to $Vy\le w$, resulting in $\ell=2p+2t+2r$. We use the $n$-fold variable bounds $\hat l,\hat u$, and the objective function $(x_0,z,x_1,\dots,x_n)\mapsto f(x_0,x_1,\dots,x_n)$ in \cref{lemma:strengthened-fixed-phase-value-function-reformulation}. We can take $\rho\le\max\{\|u_0-l_0\|_\infty,\|\sum_{i\in[n]}(u_i-l_i)\|_\infty\}\le(1+n)\|u-l\|_\infty$, implying $\log\rho=\O(\log(n)+\log\|u-l\|_\infty)$. The rest of the verification is completely analogous to the proof of \cref{thm:4-block-fpt} where $r=t$. In particular, we have $T_{\mathrm{mul}}=\O((t+s)tn)$ and $U:=\Delta tn$. Since $[A,I]\in\nfold(t,s,\Delta)$, the algorithm from~\cite{DBLP:conf/ipco/HunkenschroderKLV25} gives $T_{\mathrm{opt}}(\sigma)=(ts\Delta)^{\O(ts(t+s))}\cdot(tn+t+sn)\log(tn+t+sn)\cdot\log\sigma=(ts\Delta)^{\O(ts(t+s))}\cdot tn\log(tn)\cdot\log\sigma$ and \cref{lemma:n-fold-scaled-idp} gives $M=2^{\O((t+s)\Delta)^{t(t+s)}}$. Since we have $p'=p+t$ global variables, \cref{lemma:strengthened-fixed-phase-value-function-reformulation} shows that we can solve the 4-block IP in time
    \begin{alignat*}{2}
        &M^{p'}\cdot{}&&2^{\O(p'^2\log p')}\cdot\Bigl(\log^{\O(1)}(\rho)+\log(\rho)\\
        &&&\cdot\bigl(T_{\mathrm{opt}}((1+2U)\|u-l\|_\infty)+T_{\mathrm{mul}}+(t+sn)+\ell\bigr)\Bigr)\\
        ={}&\mathrlap{\bigl(2^{\O((t+s)\Delta)^{t(t+s)}}\bigr)^{p+t}\cdot2^{\O((p+t)^2\log(p+t))}\cdot\Bigl(\log^{\O(1)}(n)\cdot\log^{\O(1)}(\|u-l\|_\infty)}&&\\
        &&&\cdot\bigl(T_{\mathrm{opt}}((1+2\Delta tn)\|u-l\|_\infty)+\O((t+s)tn)+(t+sn)+2p+2t+2r\bigr)\Bigr)\\
        ={}&\mathrlap{2^{\O((t+s)\Delta)^{t(t+s)}\cdot p+(p+t)^2\log(p+t)}\cdot\Bigl(\log^{\O(1)}(n)\cdot\log^{\O(1)}(\|u-l\|_\infty)}&&\\
        &&&\cdot\bigl((ts\Delta)^{\O(ts(t+s))}\cdot tn\log(tn)\cdot\log((1+2\Delta tn)\|u-l\|_\infty)\\
        &&&+\O((t+s)tn)+(t+sn)+2p+2t+2r\bigr)\Bigr)\\
        ={}&\mathrlap{2^{\O((t+s)\Delta)^{t(t+s)}\cdot p+(p+t)^2\log(p+t)}\cdot(r+n)\log^{\O(1)}(n)\cdot\log^{\O(1)}\|u-l\|_\infty.}&&\tag*{\qedhere}
    \end{alignat*}
\end{proof}

\propositionnonuniformlargecoefficientfourblock*

\begin{proof}
    Instead of directly aggregating the sum of the local variables as done in the proof of \cref{proposition:uniform-large-coefficient-4-block-fpt}, we introduce $r\ell$ variables $z\in\Z^{[r]\times[\ell]}$ so that $z_{il}$ is the sum of the local variables that have the coefficient $k_l$ in the $i$-th global row.

    Again, write $l=(l_0,\hat l)$, $u=(u_0,\hat u)$ for $l_0,u_0\in\Z^p$, $\hat l,\hat u\in\Z^{tn}$, and $x=(x_0,\hat x)$ for $x_0\in\Z^p$, $\hat x\in\Z^{tn}$, and $b=(b_0,\hat b)$ for $b_0\in\Z^r$, $\hat b\in\Z^{sn}$. Aggregate $B=[B_1,\dots,B_n]$ and define the matrix $E\in\{0,1\}^{([r]\times[\ell])\times tn}$ by setting the entry in the $(i,l)$-th row and $j$-th column to $1$ if $B_{ij}=k_l$ and to $0$ otherwise.

    We apply \cref{lemma:strengthened-fixed-phase-value-function-reformulation} with the constraint matrices
    \[
        C=\left[\begin{array}{c c}
            &-I\\
            C_1&\\
            \vdots\\
            C_n&
        \end{array}\right],\text{ and }A=\left[\begin{array}{c c c}
            \multicolumn{3}{c}{E}\\
            \hdashline
            D_1&&\\
            &\ddots&\\
            &&D_n
        \end{array}\right].
    \]
    We again use the right-hand side $b'=(\veczero,\hat b)$ and write $y=(x_0,z)$. In this way, the top $r\ell$ rows encode $z_{il}=\sum_{j\in[tn]:B_{ij}=k_l}\hat x_j$. Under this constraint, the $i$-th equation of $A_0x_0+B\hat x=b_0$ is equivalent to
    \begin{align}
        (A_0x_0)_i+(B\hat x)_i&=(A_0x_0)_i+\sum_{j\in[tn]}B_{ij}\hat x_j\\
        &=(A_0x_0)_i+\sum_{l\in[\ell]}k_j\sum_{j\in[tn]:B_{ij}=k_l}\hat x_j\\
        &=(A_0x_0)_i+\sum_{l\in[\ell]}k_jz_{il}=(b_0)_i.
    \end{align}
    Since the latter left-hand side depends only on $x_0$ and $z$, we can encode this directly in $Vy\le w$ using $2r$ constraints. For $(i,l)\in[r]\times[\ell]$, we may bound the variable $z_{il}$ through
    \[
        \sum_{j\in[tn]:B_{ij}=k_l}l_j\le z_{il}\le\sum_{j\in[tn]:B_{ij}=k_l}u_j.
    \]
    We include these $2r\ell$ constraints and the global variable bounds $l_0\le x_0\le u_0$ in $Vy\le w$, resulting in $\ell'=2r+2r\ell+2p$ and $\rho\le\max\{\|u_0-l_0\|_\infty,\sum_{j\in[tn]}|u_j-l_j|\}\le(1+tn)\|u-l\|_\infty$. The rest of the proof is completely analogous to the previous proofs. We use the $n$-fold variable bounds $\hat l,\hat u$, and the objective function $(x_0,z,\hat x)\mapsto f(x_0,\hat x)$ in \cref{lemma:strengthened-fixed-phase-value-function-reformulation}. We have $T_{\mathrm{mul}}=\O((r\ell+s)tn)$ and $U:=\Delta tn$. Since $[A,I]\in\nfold(r\ell,s,\Delta)$, the algorithm from~\cite{DBLP:conf/ipco/HunkenschroderKLV25} gives $T_{\mathrm{opt}}(\sigma)=(r\ell s\Delta)^{\O(r\ell s(r\ell+s))}\cdot(tn+r\ell+sn)\log(tn+r\ell+sn)\cdot\log\sigma=(r\ell s\Delta)^{\O(r\ell s(r\ell+s))}\cdot tn\log(tn)\cdot\log\sigma$ and \cref{lemma:n-fold-scaled-idp} gives $M=2^{\O((r\ell+s)\Delta)^{r\ell(r\ell+s)}}$. Since we have $p'=p+r\ell$ global variables, \cref{lemma:strengthened-fixed-phase-value-function-reformulation} shows that we can solve the 4-block IP in time
    \begin{alignat*}{2}
        &M^{p'}\cdot{}&&2^{\O(p'^2\log p')}\cdot\Bigl(\log^{\O(1)}(\rho)+\log(\rho)\\
        &&&\cdot\bigl(T_{\mathrm{opt}}((1+2U)\|u-l\|_\infty)+T_{\mathrm{mul}}+(r\ell+sn)+\ell'\bigr)\Bigr)\\
        ={}&\mathrlap{\bigl(2^{\O((r\ell+s)\Delta)^{r\ell(r\ell+s)}}\bigr)^{p+r\ell}\cdot2^{\O((p+r\ell)^2\log(p+r\ell))}\cdot\Bigl(\log^{\O(1)}(tn)\cdot\log^{\O(1)}(\|u-l\|_\infty)}&&\\
        &&&\cdot\bigl(T_{\mathrm{opt}}((1+2\Delta tn)\|u-l\|_\infty)+\O((r\ell+s)tn)+(r\ell+sn)+2r+2r\ell+2p\bigr)\Bigr)\\
        {}&\mathrlap{2^{\O((r\ell+s)\Delta)^{r\ell(r\ell+s)}\cdot p+(p+r\ell)^2\log(p+r\ell)}\cdot\Bigl(\log^{\O(1)}(tn)\cdot\log^{\O(1)}(\|u-l\|_\infty)}&&\\
        &&&\cdot\bigl((r\ell s\Delta)^{\O(r\ell s(r\ell+s))}\cdot tn\log(tn)\cdot\log((1+2\Delta tn)\|u-l\|_\infty)\\
        &&&+\O((r\ell+s)tn)+(r\ell+sn)+2r+2r\ell+2p\bigr)\Bigr)\\
        ={}&\mathrlap{2^{\O((r\ell+s)\Delta)^{r\ell(r\ell+s)}\cdot p+(p+r\ell)^2\log(p+r\ell)}\cdot tn\log^{\O(1)}(tn)\cdot\log^{\O(1)}\|u-l\|_\infty.}&&\tag*{\qedhere}
    \end{alignat*}
\end{proof}

\propositionvariablesplustreedepth*

Let $d,\Delta\in\Z_{\ge0}$. Ligthart~\cite[Corollary 13]{ligthart2026valuefunctionsseparableconvex} shows that there exists a number $M=2^{\O(\Delta)^{2^d}}$, computable in time $\O(M)$, so that the class of matrices consisting of $A\in\{-\Delta,-\Delta+1,\dots,\Delta\}^{m\times n}$ with dual treedepth at most $d$ has the IDP after an $M$-dilation. Similarly,~\cite[Corollary 15]{ligthart2026valuefunctionsseparableconvex} provides a number $M'=2^{2^{\Delta^{2^{\O(d)}}}}$ for the matrices with primal treedepth at most $d$. In this way, the class of constraint matrices with coefficients bounded by $\Delta$ and dual or primal treedepth bounded by $d$ is seen to have the IDP after an $(M\cdot M')$-dilation.

\begin{proof}[Proof of \cref{proposition:variables-plus-treedepth}]
    Write $l=(l_0,\hat l)$ and $u=(u_0,\hat u)$ for $l_0,u_0\in\Z^p$ and $\hat l,\hat u\in\Z^n$. We apply \cref{lemma:strengthened-fixed-phase-value-function-reformulation} with the constraint matrices $C$ and $A$. We encode the global variable bounds $l_0\le y\le u_0$ in $Vy\le w$ with $\ell=2p$ constraints. We directly insert the right hand side vector $b$, the second-stage variable bounds $\hat l,\hat u$, and the objective function $(y,x)\mapsto f(y,x)$ into \cref{lemma:strengthened-fixed-phase-value-function-reformulation}. We can take $\rho\le\|u_0-l_0\|_\infty\le\|u-l\|_\infty$. Since the number of nonzero entries of $A$ is at most $dn$ (in the case of bounded dual treedepth) or $dm$ (in the case of bounded primal treedepth), products $Ax$ can be computed in time $T_{\mathrm{mul}}=\O(d(n+m))$. We can additionally estimate $\|A\|_\infty\le\Delta n=:U$. The primal or dual treedepth of $[A,I]$ is at most one greater than the primal or dual treedepth of $A$ respectively. Therefore, we can solve (\ref{eq:generic-second-stage-ip}) using one of the two algorithms from~\cite{DBLP:conf/ipco/HunkenschroderKLV25}, which run within time $T_{\mathrm{opt}}(\sigma)=f_1(d+1,\Delta)\cdot(n+m)\log(n+m)\cdot\log\sigma$ for some computable function $f_1$. By taking $M=2^{\O(\Delta)^{2^{d+1}}}\cdot2^{2^{\Delta^{2^{\O(d+1)}}}}$ to be the product of the numbers from~\cite{ligthart2026valuefunctionsseparableconvex} for $d+1$, we find that $[A,I]$ has the IDP after an $M$-dilation when $A$ has treedepth bounded by $d$ and coefficients with an absolute value of at most $\Delta$. Since this class of matrices is closed under inverting columns, \cref{lemma:idp-iff-convex-extensibility-along-line} shows that (\ref{eq:value-function-assumed-convex-extensibility}) is convex extensible on any line, which again implies that it is integer midpoint convex. Now \cref{lemma:strengthened-fixed-phase-value-function-reformulation} shows that we can solve the IP in time
    \begin{alignat*}{2}
        &M^p\cdot{}&&2^{\O(p^2\log p)}\cdot\Bigl(\log^{\O(1)}(\rho)+\log(\rho)\cdot\bigl(T_{\mathrm{opt}}((1+2U)\|u-l\|_\infty)+T_{\mathrm{mul}}+m+\ell\bigr)\Bigr)\\
        ={}&\mathrlap{f_2(p,d,\Delta)\cdot\Bigl(\log^{\O(1)}(\|u-l\|_\infty)}&&\\
        &&&\cdot\bigl(T_{\mathrm{opt}}((1+2\Delta n)\|u-l\|_\infty)+\O(d(n+m))+m+2p\bigr)\Bigr)\\
        ={}&\mathrlap{f_2(p,d,\Delta)\cdot\Bigl(\log^{\O(1)}(\|u-l\|_\infty)\cdot\bigl(f_1(d+1,\Delta)\cdot(n+m)\log(n+m)}&&\\
        &&&\cdot\log((1+2\Delta n)\|u-l\|_\infty)+\O(d(n+m)+p)\bigr)\Bigr)\\
        \le{}&\mathrlap{f_3(p,d,\Delta)\cdot\log^{\O(1)}(\|u-l\|_\infty)\cdot\bigl((n+m)\log(n+m)}&&\\
        &&&\cdot\bigl(\log(\Delta)+\log(n)+\log(\|u-l\|_\infty)\bigr)+\O(n+m)\bigr)\\
        \le{}&\mathrlap{f_4(p,d,\Delta)\cdot(n+m)\log^{\O(1)}(n+m)\cdot\log^{\O(1)}\|u-l\|_\infty,}&&
    \end{alignat*}
    for some computable functions $f_2$, $f_3$, and $f_4$.
\end{proof}

\propositionboundedfractionalityisinsufficient*

\begin{proof}
    Let $k=2T+1$, which is odd. Let $S\in\{0,1\}^{k\times k}$ be the permutation matrix given by $S_{ij}=1$ if $i\equiv j+1\pmod k$ and $S_{ij}=0$ otherwise. Set $Q=\sum_{i=0}^{k-1}(-S)^i\in\{-1,1\}^{k\times k}$ and observe that
    \[
        Q(I+S)=\Bigl(\sum_{i=0}^{k-1}(-S)^i\Bigr)(S^0+S^1)=\sum_{i=0}^{k-1}(-S)^i-\sum_{i=1}^k(-S)^i=S^0+S^k=2I.
    \]
    We let $A\in\{-1,0,1\}^{(k^2+1)\times2k}$ be the matrix
    \[
        \left[\begin{array}{c c}
            Q & -I\\
            \vdots & \vdots\\
            Q & -I\\
            \vecone^\top &
        \end{array}\right]
        \begin{array}{l}
            \left.
            \begin{array}{@{}l@{}}
                \vphantom{Q} \\
                \vphantom{\vdots} \\
                \vphantom{Q}
            \end{array}
            \right\} k \text{ times,} \\
            \vphantom{\vecone^\top}
        \end{array}
    \]
    where $\vecone$ is the all-ones vector.

    To show that $P_b$ is $\tfrac12$-integral for all integral $b$, it suffices to show that any nonsingular square submatrix of $H=[Q;\vecone^\top]\in\{-1,0,1\}^{(k+1)\times k}$ has an inverse that is $\tfrac12$-integral. Here, $;$ denotes vertical concatenation. Let $D\subseteq[2k+1]$ be an arbitrary index set of cardinality $k$ that indexes the rows of $[H;I]\in\Z^{(2k+1)\times k}$ and let $c\in\Z^k$ be arbitrary. It suffices to show that if $\overline B$ is the submatrix of $[H;I]$ restricted to the rows in $D$ and if $\overline B$ is nonsingular, then the unique solution $x$ to $\overline Bx=c$ is $\tfrac12$-integral. The case where $D=\{k+2,\dots,2k+1\}$, i.e., $\overline B=I$, is trivial. In the other case, we substitute $y=2(I+S)^{-1}x$ and show that there is an integral solution $y$, which implies that the unique solution $x=\tfrac12(I+S)y$ must be $\tfrac12$-integral. The $i$-th row of the system $\tfrac12\overline B(I+S)y=c$, which is equivalent to $\overline B(I+S)y=2c$, reads:
    \begin{itemize}
        \item $(Q(I+S))_{i\cdot}y=(2I)_{i\cdot}=2y_i=2c_i\iff y_i=c_i$ if $i\in\{1,\dots,k\}$;
        \item $\vecone^\top(I+S)y=2\vecone^\top y=2c_i\iff\vecone^\top y=c_i$ if $i=k+1$; and
        \item $(I(I+S))_{(i-(k-1))\cdot}y=(I+S)_{(i-(k-1))\cdot}y=2c_i$ if $i\in\{k+2,\dots,2k+1\}$;
    \end{itemize}
    Here, $C_{i\cdot}$ denotes the $i$-th row of a matrix $C$. The first type of constraint fixes some variables to integers. For this reason, it is sufficient to argue that the determinant of an arbitrary nonsingular submatrix $\hat B$ of $[\vecone^\top;I+S]$ (unequal to $I+S$) is $-1$ or $1$. Let $\hat B'$ be the submatrix of $\hat B$ by removing the row corresponding to $\vecone^\top$, if it is present. Since $\hat B'$ does not contain all of the rows of $I+S$, it is a submatrix of the vertex-edge incidence matrix of a union of path graphs. Therefore, we can permute its columns so that the ones in every row are consecutive. This means that (a permutation of) $\hat B$ is (the transpose of) an interval matrix and is totally unimodular~\cite{DBLP:books/daglib/0090562}, completing the argument.

    With regards to the IDP, we choose $b^*=(\veczero,1)$. Let $M\in[T]$. The integral point $(x^*,u^*,s^*)$ given by $x^*=\vecone\in\Z^k$, $u^*=\vecone\in\Z^k$, and $s^*=(\veczero,k\cdot M-k)$, is in the dilated polyhedron $k\cdot M\cdot P_b$ since $Q\vecone=\vecone$. Suppose for a contradiction that $M\cdot P_b$ has the IDP and consider an arbitrary point $(x,u,s)$ in $M\cdot P_b$ with $x,u\in\{0,1\}^k$ and $s\in\{\veczero\}\times\Z_{\ge0}$, and $x$ being nonzero, which would have to appear in a decomposition of $(x^*,u^*,s^*)$. This point satisfies $Qx=u$, implying $x=\tfrac12(I+S)u$, i.e., $x_i=\tfrac12u_i+\tfrac12u_{(i+1)\bmod k}$ for $i\in[k]$. Since $x$ and $u$ are binary, this fixes $x_i=u_i=u_{(i+1)\bmod k}$ for all $i\in[k]$. Hence $(x,u)\in\{(\veczero,\veczero),(\vecone,\vecone)\}$. Since $x$ was nonzero, the only option is $x=u=\vecone$. However, now the last component of $s$ equals $M-\vecone^\top x=M-k\le T-k=-1-T<0$, which is the required contradiction.

    With respect to the value function $h$, we simplify
    \begin{align*}
        h(t)&=\min\bigl\{\|s\|_1\bigm\vert Az+s=(\veczero,t),\,z\in\{0,1\}^n,\,s\in\Z^m\bigr\}\\
        &=\min\bigl\{\|Az-(\veczero,t)\|_1\bigm\vert z\in\{0,1\}^n\bigr\}\\
        &=\min\bigl\{k\cdot\|Qx-u\|_1+|\vecone^\top x-t|\bigm\vert\ x\in\{0,1\}^k,\,u\in\{0,1\}^k\bigr\}.
    \end{align*}
    The solution $x=u=\vecone$ satisfies $Qx-u=\vecone-\vecone=\veczero$ and shows that $h(t)\le|k-t|$. Similarly, the solution $x=u=\veczero$ also satisfies $Qx-u=\veczero$, showing that $h(t)\le|t|$. Consider any other solution $(x,u)\notin\{(\veczero,\veczero),(\vecone,\vecone)\}$. Following the previous argument, $Qx-u$ cannot be zero, which implies that $\|Qx-u\|_1\ge1$. This shows that such solution has an objective value of at least $k+|\vecone^\top x-t|$, which is at least as large as $|k-t|$ or $|t|$. We can conclude that $h(t)=\min\{|2T+1-t|,|t|\}$. Now, choosing $r^*=T$ shows that for any $M\in[T]$ it holds that
    \[
        h(T)=T>T+\tfrac12-M=\tfrac12|T-M|+\tfrac12|2T+1-(T+M)|=\tfrac12h(T-M)+\tfrac12h(T+M).\tag*{\qedhere}
    \]
\end{proof}

%% file: sections/applications.tex
\section{Application to Max-Min fair many-to-one matching with size constraints}
\label{sec:applications}

We give an application of \cref{thm:4-block-fpt} to matching under fairness constraints by using its natural integer programming formulation. This resolves an open question of Boehmer and Koana~\cite{DBLP:journals/talg/BoehmerK24}. The problem they study is the following. The input is a bipartite graph~$G=(U\mathbin{\dot\cup}V,E)$, a set $\mathcal C$ of colors with a coloring function~$\mathrm{col}\colon U\to\mathcal C$, and an integer $\ell\ge0$. A \emph{left-perfect many-to-one matching} is a set $M\subseteq E$ that contains exactly one edge incident to each $u\in U$. For $v\in V$ and $c\in\mathcal C$, let $M(v)$ be the set of vertices matched to $v$ and $M(v)_c$ the set of those of color $c$. The matching is \emph{$\ell$-fair} (with respect to Max-Min) if
\[
    \max_{c\in\mathcal C}|M(v)_c|-\min_{c\in\mathcal C}|M(v)_c|\le\ell\qquad\text{for all }v\in V,
\]
where both extrema range over all of $\mathcal C$, so a color absent from $M(v)$ counts with $0$. In the Max-Min Fair Matching with size constraints problem, two integers $q^-,q^+$ are additionally given, and the question is whether there is an $\ell$-fair matching $M$ with $q^-\le|M(v)|\le q^+$ for all $v\in V$.

Let $k=|V|$. Boehmer and Koana show that the problem is FPT parameterized by $k+|\mathcal C|$~\cite[Theorem 3.3]{DBLP:journals/talg/BoehmerK24}, via an ILP with few variables and Lenstra's algorithm, and that without size constraints, and also with the non-emptiness constraint $|M(v)|\ge1$, it is FPT parameterized by $k$ alone~\cite[Section 4]{DBLP:journals/talg/BoehmerK24}. They ask whether this extends to arbitrary size constraints~\cite[Section 7]{DBLP:journals/talg/BoehmerK24}: \emph{``does Fair Matching remain fixed-parameter tractable with respect to $k$?''} We answer this for Max-Min fairness. The question also concerns a second fairness measure, the margin of victory, which our formulation does not cover; we believe that this is because the problem is actually hard in this second case.

\begin{theorem}
    \label{thm:fair-matching}
    The Max-Min Fair Matching with size constraints problem can be solved in time $f(k)\cdot|I|^{\O(1)}$, where $k=|V|$, $|I|$ is the encoding length of the instance, and $f$ is a computable function.
\end{theorem}

We first discuss the natural integer programming formulation of the problem. For $c\in\mathcal C$ and a set $T\subseteq V$, let $N_{c,T}$ be the number of vertices $u\in U$ with $\mathrm{col}(u)=c$ and neighborhood $N_G(u)=T$. Vertices with the same color and the same neighborhood are interchangeable, so it suffices to decide, for every $v\in T$, the number $a_{c,T,v}$ of such vertices matched to $v$. The fairness condition at $v$ says that all color counts at $v$ lie in a common interval $[\alpha_v,\alpha_v+\ell]$ for some integer $\alpha_v$, and we make $\alpha_v$ a variable. With sums over $T$ ranging over all subsets of $V$, the IP is
\begin{alignat}{2}
    \textstyle\sum_{c\in\mathcal C}\sum_{T\ni v}a_{c,T,v}-L_v&=0,&\qquad&v\in V,\label{eq:fm-load}\\
    \textstyle\sum_{v\in T}a_{c,T,v}&=N_{c,T},&&c\in\mathcal C,\ T\subseteq V,\label{eq:fm-mass}\\
    \textstyle\sum_{T\ni v}a_{c,T,v}-\alpha_v-\sigma_{c,v}&=0,&&c\in\mathcal C,\ v\in V,\label{eq:fm-fair}
\end{alignat}
with integer variables bounded by $0\le\alpha_v\le|U|$, $q^-\le L_v\le q^+$, $0\le a_{c,T,v}\le N_{c,T}$, and $0\le\sigma_{c,v}\le\ell$, and with the zero objective. Constraints~(\ref{eq:fm-load}) with the bounds on $L_v$ are the size constraints, and Constraints~(\ref{eq:fm-fair}) with the bounds on $\sigma_{c,v}$ say $\alpha_v\le|M(v)_c|\le\alpha_v+\ell$.

This is a 4-block IP with one block per color, so $n=|\mathcal C|$. The global variables are $(\alpha_v)_{v\in V}$ and $(L_v)_{v\in V}$, so $p=2k$. The local variables of block $c$ are the $a_{c,T,v}$ for $v\in T$ and the $\sigma_{c,v}$, so $t=k2^{k-1}+k$. The global constraints are~(\ref{eq:fm-load}), so $r=k$, and the local constraints of block $c$ are~(\ref{eq:fm-mass}) and~(\ref{eq:fm-fair}), so $s=2^k+k$. Explicitly,
\[
    A_0=\begin{bmatrix}\matzero&-I\end{bmatrix},\quad
    B_c=\begin{bmatrix}Q&\matzero\end{bmatrix},\quad
    C_c=\begin{bmatrix}\matzero&\matzero\\-I&\matzero\end{bmatrix},\quad
    D_c=\begin{bmatrix}W&\matzero\\Q&-I\end{bmatrix},
\]
where $Q\in\{0,1\}^{k\times k2^{k-1}}$ maps $(a_{c,T,v})_{T,v}$ to $(\sum_{T\ni v}a_{c,T,v})_{v\in V}$ and $W\in\{0,1\}^{2^k\times k2^{k-1}}$ maps it to $(\sum_{v\in T}a_{c,T,v})_{T}$. All entries are in $\{-1,0,1\}$, so $\Delta=1$; the multiplicities $N_{c,T}$, the bounds $q^-,q^+$ and the tolerance $\ell$ occur only in right-hand sides and variable bounds.

The same formulation also handles recipient-specific bounds $q^-_v\le|M(v)|\le q^+_v$ and tolerances $\ell_v$, by changing the bounds on $L_v$ and $\sigma_{c,v}$.

\begin{proof}[Proof of \cref{thm:fair-matching}]
    All loads $|M(v)|$ and color counts $|M(v)_c|$ lie in $\{0,\dots,|U|\}$, so we may replace $q^-$ by $\max\{q^-,0\}$, $q^+$ by $\min\{q^+,|U|\}$, and $\ell$ by $\min\{\ell,|U|\}$ without changing the answer. We compute the numbers $N_{c,T}$ in time $\O(k|U|+2^k|\mathcal C|)$ and set up the IP (\ref{eq:fm-load})--(\ref{eq:fm-fair}) with the earlier specified variable bounds. We show that it is feasible if and only if there is an $\ell$-fair matching $M$ with $q^-\le|M(v)|\le q^+$ for all $v\in V$.

    Let $M$ be such a matching. Let $a_{c,T,v}$ be the number of vertices $u$ with $\mathrm{col}(u)=c$ and $N_G(u)=T$ that are matched to $v$; since $u$ is matched to a neighbor, $a_{c,T,v}=0$ unless $v\in T$, so these variables satisfy (\ref{eq:fm-mass}) and $0\le a_{c,T,v}\le N_{c,T}$. Moreover, $\sum_{T\ni v}a_{c,T,v}=|M(v)_c|$. Set $L_v=|M(v)|$, which satisfies (\ref{eq:fm-load}) and $q^-\le L_v\le q^+$. Set $\alpha_v=\min_{c\in\mathcal C}|M(v)_c|\in\{0,\dots,|U|\}$ and $\sigma_{c,v}=|M(v)_c|-\alpha_v$, which satisfies (\ref{eq:fm-fair}), and $0\le\sigma_{c,v}\le\ell$ holds because $M$ is $\ell$-fair.

    Conversely, take an integral feasible solution. For every color $c$ and $T\subseteq V$, partition the $N_{c,T}$ vertices of color $c$ and neighborhood $T$ into sets of sizes $a_{c,T,v}$ for $v\in T$, which is possible by (\ref{eq:fm-mass}) and the nonnegativity of the $a_{c,T,v}$, and match the vertices of the set for $v$ to $v$. Every vertex of $U$ is matched exactly once, and along an edge of $G$, since $v\in T$. The resulting matching $M$ satisfies $|M(v)_c|=\sum_{T\ni v}a_{c,T,v}$. By (\ref{eq:fm-load}), $|M(v)|=L_v\in[q^-,q^+]$, and by (\ref{eq:fm-fair}), $|M(v)_c|=\alpha_v+\sigma_{c,v}\in[\alpha_v,\alpha_v+\ell]$ for every $c\in\mathcal C$. Hence, $\max_{c}|M(v)_c|-\min_{c}|M(v)_c|\le\ell$ for every $v\in V$. This construction takes time $\O(|U|+k2^k|\mathcal C|)$.

    As shown above, the IP is a 4-block IP with $n=|\mathcal C|$, $p=2k$, $r=k$, $s=2^k+k$, $t=k2^{k-1}+k$, and $\Delta=1$, and after the normalization above $\|u-l\|_\infty\le|U|$. \Cref{thm:4-block-fpt} with the zero objective solves it in time
    \begin{align*}
        &2^{\O(2^k+2k)^{k(2^k+2k)}\cdot2k+\O(k^2\log k)}\cdot tn\log^{\O(1)}(tn)\cdot\log^{\O(1)}|U|\\
        &=2^{2^{2^{\O(k)}}}\cdot|\mathcal C|\log^{\O(1)}(|\mathcal C|)\cdot\log^{\O(1)}(|U|).
    \end{align*}
    Together with the preprocessing and the construction of $M$, this gives the claimed running time.
\end{proof}

The used natural IP formulation requires both coupling global variables $\alpha_v$ and coupling global constraints (\ref{eq:fm-load}), which does not naturally fit inside either a 2-stage stochastic IP nor an $n$-fold IP. Additionally, the global variables have domains whose size is not bounded by a function of the parameter, so they cannot simply be guessed. We note that, without size constraints, (\ref{eq:fm-load}) and the variables $L_v$ can be dropped, and what remains is a 2-stage stochastic IP; this is the case that was already known to be FPT in $k$.